\documentclass{eptcs}
\providecommand{\event}{ICLP 2019}
\usepackage{underscore}
\usepackage[utf8]{inputenc}
\usepackage{amsmath}
\usepackage{amsthm}
\usepackage{bm}
\usepackage{semantic}
\usepackage{nicefrac}
\usepackage{marginnote}

\def\appendix{\par
\section*{APPENDIX}
\setcounter{section}{0}
\setcounter{subsection}{0}
\setcounter{lmm}{0}
\setcounter{thm}{0}
\def\thesection{\Alph{section}} }

\newcounter{thm}
\newcounter{lmm}
\newcounter{coro}
\newcounter{propo}
\newcounter{defin}
\newcounter{exmpl}

\newenvironment{theorem}[1][]{\par\vspace{\baselineskip}\noindent\refstepcounter{thm}\textbf{Theorem \thethm #1:}}{\\ \ignorespacesafterend}

\newenvironment{lemma}[1][]{\par\vspace{\baselineskip}\noindent\refstepcounter{lmm}\textbf{Lemma \thelmm #1:}}{\\ \ignorespacesafterend}

\newenvironment{defn}[1][]{\par\vspace{\baselineskip}\noindent\refstepcounter{defin}\textbf{Definition \thedefin #1:}}{}
\newenvironment{example}{\par\vspace{\baselineskip}\noindent\refstepcounter{exmpl}\textbf{Example \theexmpl:}}{}

\title{A Three-Valued Semantics for Typed Logic Programming}
\author{João Barbosa \quad\quad Mário Florido \quad\quad Vítor Santos Costa \institute{Faculty of Science of the University of Porto, Portugal}}

\def\titlerunning{A Three-Valued Semantics for Types Logic Programming}
\def\authorrunning{J. Barbosa, M. Florido \& V. Santos Costa}

\begin{document}

\maketitle

\begin{abstract}
Types in logic programming have focused on conservative approximations of program semantics by regular types, on one hand, and on type systems based on a prescriptive semantics defined for typed programs, on the other. In this paper, we define a new semantics for logic programming, where programs evaluate to true, false, and to a new semantic value called \emph{wrong}, corresponding to a run-time type error. We then have a type language with a separated semantics of types. Finally, we define a type system for logic programming and prove that it is semantically sound with respect to a semantic relation between programs and types where, if a program has a type, then its semantics is not \textit{wrong}. Our work follows Milner's approach for typed functional languages where the semantics of programs is independent from the semantic of types, and the type system is proved to be sound with respect to a relation between both semantics.
\end{abstract}

\section{Introduction}

Many-valued logics, in which there are more than two truth values, have been defined and applied to several different scenarios. As examples, consider Kleene three-valued logic $K_3$ \cite{GlossarWiki:Kleene:1938}, which includes a third semantic value corresponding to {\em undefined}, and Bochvar's internal three-valued logic ($B_{3}^{I}$) \cite{3value}, also known as Kleene's weak three-valued logic \cite{klccne1952im}, where the intermediate truth value corresponds to {\em nonsense}. The third value of Bochvar's internal logic was reinterpreted by Beall \cite{AJL} as {\em off-topic}.

Within the context of programming languages, we can observe that executing a program is meaningful when values are within their expected semantic domains. This suggests that Bochvar's third value would match ill-typed programs. We argue that this formalism is particularly suited to logic programs. First, using the third value for run-time type errors allows one to distinguish a program that simply has no solution and a program that combines arguments in predicate calls haphazardly. Second, Bochvar's logic naturally propagates the errors. Based on this observation, we present a novel application of Bochvar's internal three-valued logic, to the definition of a new semantics for logic programming, where Bochvar's {\em nonsense} value stands for a type error at run-time. Inspired by Milner's motto “Well-typed programs can’t go wrong” \cite{DBLP:journals/jcss/Milner78}, our reinterpretation of Bochvar's third value corresponds to Milner's notion of {\em wrong}. 

Our three-valued semantics uses a set of disjoint semantic domains. Given that unification is the basic execution mechanism of logic programming, we define a run-time type error as an attempt to unify terms belonging to different semantic domains or applying functions or predicates to terms whose values are not on their domains.

In this view of logic programming, something may be false and still make sense (well-typed), while something may be true, and yet nonsensical (ill-typed).

\begin{example} \label{ex1}
Let $p$ be the predicate defined as follows:
\begin{verbatim}
p(X) :- X = 1, X = 2.
\end{verbatim}

If the constants 1 and 2 belong to the same semantic domain (intuitively, some set of integers) then the predicate should be considered well-typed, in the sense that it will not produce a run-time type error. Top-down execution in Prolog or a typed extension of Prolog just fails.
Now suppose that we are using the Herbrand interpretation where every semantic value belongs to a single domain: the set of all ground terms, also known as the Herbrand universe. In this case there is also no type error in the program and the result of executing it with any query would be \textit{false}. On the other hand, suppose that the domain of interpretation consists of singleton domains, each one containing exactly one different semantic value. In this case, 1 and 2 would belong to different semantic domains and the program would generate a type error at run-time for any query, since if $X \neq 1$ a type error would be reported on the first unification and if $X = 1$, a type error would be reported on the second unification.
\end{example}

\begin{example} \label{ex2}
Let $q$ be a predicate defined as follows:
\begin{verbatim}
q(X) :- X = 1, X = a.
\end{verbatim}

With this example we have the same situation that we had on the previous one, where the predicate does not generate run-time type errors if there is only one domain containing all ground terms and it would generate a type error at run-time if we have several singleton domains, each one containing one value. 
But note that, in practice, we would use a set of domains which is more informative, i.e. it is not one of the previously mentioned two extremes. In that case, one of these programs could generate run-time type errors while the other could not. 
If our semantic domains correspond to types as they are usually used in programming languages, then all integer numbers would belong to the same domain, as would all floating-point numbers, characters and atomic constants, each belong to their own domain. In this scenario, the program in example \ref{ex1}, although it fails, will not originate a type error at run-time. On the other hand, the program in example \ref{ex2} produces an error for any query, because, it will try to unify values belonging to different semantic domains, in this case an integer and an atomic constant.
\end{example}\\

This view of logic programming, where queries may succeed, but generate type errors at run-time, is compatible with the notion of {\em intended interpretation} defined in \cite{DBLP:books/mit/pfenning92/Naish92,BarbosaFloridoCosta17}.
A typical example of a Prolog predicate which may not follow the programmer's {\em intended interpretation} is the $append$ predicate:

\begin{verbatim}
append([],X,X).
append([X|L],Y,[X|L1]) :- append(L,Y,L1).
\end{verbatim}

This predicate succeeds for the query $append([~],1,1)$ without type checking, but this query does not makes sense if the type implicitly intended by the programmer for the predicate is the expected one, which is $list(\alpha) \times list(\alpha) \times list(\alpha)$, meaning three polymorphic lists.

\paragraph{}
The major novel contributions of our work are:
\begin{enumerate}
\item A new semantics for logic programming based on a three-valued logic which captures run-time type errors, independently of failure or success;
\item A new type system for logic programming, which relates programs with types, defining statically a well-typing relation;
\item A proof that our type system is sound with respect to the above-mentioned semantics and that the semantics of well-typed logic programs cannot be {\em wrong}. 
\end{enumerate}

Previous work on declarative semantics for logic programming is based on model theory, where some interpretation function either makes clauses true or false in a certain domain \cite{VanEmden:1976:SPL:321978.321991}. Here we extend this model theory approach by using a three-valued logic.
Using a specific semantic value for denoting erroneous programs goes back to early work on declarative debugging in logic programming \cite{Pereira:1986:RDL:645520.656148}. In these early works whether inadmissible atoms succeed or fail was not important. 
This idea was further formalized through making explicit use of the third semantic value, in a previously defined three-valued semantics for logic programming \cite{DBLP:journals/tplp/Naish06}. This previous semantics was based on a generalization of the $T_P$ operator for the {\em strong} Kleene logic, which captured inadmissibility with respect to a specification containing mode and type information. The main differences to our work is that we use the {\em weak} Kleene logic to denote the propagating effect of type errors, which enables us to use our semantics to establish the {\em semantic soundness} of a type system for logic programming, using the third semantic value as the interpretation of {\em ill-typed} atoms.
Previous semantics for typed logic programming, based on different domains of interpretation, were defined before using many-sorted logics \cite{DBLP:conf/slp/LakshmanR91,DBLP:books/daglib/0095081}. These semantics were defined for languages where type declarations formed an integral part of program syntax and were also used to determine their semantics. Our work differs from these approaches by defining separate semantics for untyped programs and types. Three-valued domains of interpretation revealed to be crucial in this separation of both semantics.

In several previous works types approximated the success set of a predicate \cite{DBLP:conf/iclp/Zobel87,DBLP:books/mit/pfenning92/DartZ92,DBLP:books/mit/pfenning92/YardeniFS92,DBLP:conf/iclp/BruynoogheJ88,DBLP:conf/lics/FruhwirthSVY91}. This sometimes led to overly broad and even useless types, because the way logic programs are written can be very general and accept more than what was initially intended. These approaches were different from ours in the sense that in our work types can filter the set of terms accepted by a predicate.
A different approach relied on ideas coming from functional programming languages \cite{DBLP:journals/ai/MycroftO84,DBLP:conf/slp/LakshmanR91,DBLP:books/daglib/0095081,DBLP:conf/iclp/SchrijversCWD08}. Other examples of the influence of functional languages on types for logic programming are the prescriptive type systems used in several functional logic programming languages \cite{DBLP:conf/birthday/Hanus13,DBLP:journals/jlp/SomogyiHC96}.
Along this line of research, a rather influential type system for logic programs was Mycroft and O'Keefe type system \cite{DBLP:journals/ai/MycroftO84}, which was later reconstructed by Lakshman and Reddy \cite{DBLP:conf/slp/LakshmanR91}. This system had types declared for the constants, function symbols and predicate symbols used in a program. The semantics of these systems and our semantics are quite different. Their semantics was itself typed, while our semantics uses independent semantics for programs and types. Another difference is the type language used in these previous works and in our work. In the Mycroft-O'Keefe type system, each clause of a predicate must have the same type. We lift this limitation extending the type language with sums of types (union types), where the type of a predicate is the sum of the types of its clauses.
Other relevant works on type systems and type inference in logic programming include types used in the logic programming systems CIAO Prolog \cite{10.1007/BFb0026840,Vaucheret:2002:MPY:647171.718317,DBLP:journals/tplp/HermenegildoBCLMMP12}, SWI and Yap \cite{DBLP:conf/iclp/SchrijversCWD08}. These systems were dedicated to the type inference problem and are not based on a declarative semantics with an explicit notion of type error.

\section{A Three-Valued Semantics for Logic Programming}

In this section we present a declarative semantics for logic programming based on the three-valued logic defined by Kleene \cite{klccne1952im} and further interpreted by Bochvar \cite{3value} and Beall \cite{AJL}. The third logic value was named \textit{nonsense} in Bochvar's work and interpreted as a meaningless statement that spreads its meaninglessness through every connective. This is the reason why whenever a connective joins \textit{nonsense} with any formula, the result is always \textit{nonsense}. We can notice that there is a similarity between this behaviour and the propagation of run-time errors in a programming language.
 Particularly in logic programming, the use of this third semantic value for run-time type errors allows one to distinguish a program that simply fails from a program that erroneously uses its function and predicate arguments.
This change in the semantics of logic programming from a two-valued semantics to a three-valued semantics captures the notion of type error and type-safeness and thus it will be the key in establishing the precise meaning of what is a {\em semantically sound type system} for logic programming.
We will consider that the semantics of our programs follows a three-valued logic, where the values are \emph{true}, \emph{false} and \emph{wrong}. The description of the connectives in this three-valued logic is described in table \ref{con}.

\label{threeValued}

\begin{table}[!htb]
    \begin{minipage}{.50\linewidth}
        \begin{center}
        \begin{tabular}{| c | c | c | c |}
            \hline
            $\bm{\wedge}$ & \textit{true} & \textit{false} & \emph{wrong} \\ \hline
            \textit{true} & \textit{true} & \textit{false} & \emph{wrong} \\ \hline
            \textit{false} & \textit{false} & \textit{false} & \emph{wrong} \\ \hline
            \emph{wrong} & \emph{wrong} & \emph{wrong} & \emph{wrong}  \\ \hline
        \end{tabular}
        \end{center}
    \end{minipage}%
    \begin{minipage}{.50\linewidth}
        \begin{center}
        \begin{tabular}{| c | c | c | c |}
            \hline
            $\bm{\vee}$ & \textit{true} & \textit{false} & \textit{wrong} \\ \hline
            \textit{true} & \textit{true} & \textit{true} & \textit{wrong} \\ \hline
            \textit{false} & \textit{true} & \textit{false} & \textit{wrong} \\ \hline
            \textit{wrong} & \textit{wrong} & \textit{wrong} & \textit{wrong}  \\ \hline
            \end{tabular}
            \end{center}
    \end{minipage}%
\caption{Connectives of the three-valued logic - conjunction and disjunction}
\label{con}
\end{table}

The negation of logic values is defined as: $\neg true = false$, $\neg false = true$ and $\neg \emph{wrong} = \emph{wrong}$. And implication is defined as: $p\implies q \equiv (\neg p) \vee q$.

To understand the meaning of logic programs, here we define their semantics. We define a declarative semantics, in the sense that it explains \textit{what} logic programs compute, but not \textit{how} logic programs compute. So it abstracts from the details of the computation and focuses on the logical meaning of predicates, when interpreted in the logic described above. As we shall see later, this will significantly simplify the task of defining a sound type system with respect to a semantic relation between programs and types.

{\bf Normalized Programs:} in order to simplify further processing we reduce every predicate to a {\em normal form} \cite{VanRoy:1991:LPE:128589}. In this representation, each predicate is defined by a single clause ($H :- B$), where the head $H$ contains distinct variables as arguments and the body $B$ is a disjunction (represented by the symbol $;$) of sequences of goals. We assume that there are no common variables between sequences of goals, except for the variables that occur in the head of the clause, without loss of generality. 

\begin{example}\label{ex3}
Let $add$ be a predicate defined by:
\begin{verbatim}
add(0,X,X).
add(s(X),Y,s(Z)) :- add(X,Y,Z).
\end{verbatim}
\noindent The normal form of this predicate is:
\begin{verbatim}
add(X1,X2,X3) :- ( X1 = 0, X2 = X, X3 = X ) ;
                 ( X1 = s(X'), X2 = Y, X3 = s(Z),
                 X4 = X', X5 = Y, X6 = Z, add(X4,X5,X6) ) .
\end{verbatim}
\end{example}

Note that it is always possible to normalize a program using program transformation \cite{VanRoy:1991:LPE:128589}. In the rest of this paper we will assume that predicate definitions are always in normal form.

{\bf Semantics:} we assume the existence of several distinct domains of interpretation, consisting of non-empty sets of semantic values. We include a singleton set $\mathbf{W}$ containing a value \textit{wrong} as the only element, corresponding to a type error at run-time. Errors are generated dynamically by trying to unify terms from different domains, or applying functions or predicates to terms whose values are not on their domains.  

Let \textbf{Val} be the set of semantic values of our term language. We will consider a disjoint union of domains:
\begin{center}
$\mathbf{Val} = \mathbf{B_1} + \dots + \mathbf{B_n} + \mathbf{A_1} + \dots + \mathbf{A_m} + \mathbf{F} + \mathbf{Bool} + \mathbf{W}$
\end{center}
Domains $\mathbf{B_i}$ are called {\em basic domains} and are the domains of constant symbols. Domains $\mathbf{A_i}$ are the {\em tree domains} of semantic values consisting of finite trees. The distinction between them is purely from their nature and not from their functionality or role in the semantics. \textbf{Bool} is the domain which contains values $true$ and $false$. $\mathbf{F}$ is the domain of all semantic functions, such that each semantic function maps a tuple of basic or tree domains, whose arity is the arity of the function, into a basic domain, a tree domain or \textbf{Bool}. 
A domain will be represented throughout the paper by $D$ or $D_i$, for some $i$.

Each ground term in our language will be associated with a semantic value, contained in a semantic domain, by a semantic interpretation function. Let \textbf{Var} be an infinite and enumerable set of variables, \textbf{Func} be an infinite and enumerable set of function symbols and \textbf{Pred} be an infinite and enumerable set of predicate symbols.

There is an interpretation function $I:\textbf{Func} \cap \textbf{Pred}\to \textbf{Val}$ which associates each constant to their semantic value and each function symbol and predicate symbol to a functional value $f$. Functions associated with predicate symbols always have the domain \textbf{Bool} as co-domain. Note that, because $I$ is a function, for a given $I$, each constant, function symbol and predicate symbol can only be associated with one semantic value. Therefore, since domains are disjoint, meaning $\forall i,j. i \neq j \implies D_i \cap D_j = \emptyset$, each semantic value belongs to a unique domain. From now on we will write $domain(v)$ to denote the domain which contains $v$, if $v$ is not a function, or to denote the tuple with the domains of the arguments of $v$, if $v$ is a function.

We shall now introduce the concept of a {\em state} which will associate variables with their current value. We shall represent a state as a function from variables to values. Each state $\sigma$ specifies a value, written $\sigma(X)$, for each variable $X$ of $\textbf{Var}$. 

Note that in clauses in normal form $q(X_1,\dots,X_n) :- sg_1;\dots;sg_m.$, the same variable symbols $X_1,\dots,X_n$ are used in the body $sg_1;\dots;sg_m$ but denote different possible values in the different sequences of goals $sg_1;\dots;sg_m$. Thus, to define its semantics, we will need a list of possible different states. Each state $\sigma_1,\dots,\sigma_m$ in the definition of the semantics of a clause $|[ q(X_1,\dots,X_n) :- sg_1;\dots;sg_m.|]_{I,[\sigma_1,\dots,\sigma_m]}$ will correspond to different variants of variables $X_1,\dots,X_n$, one for each query $sg_i$, for $1 \leq i \leq m$.

For simplicity of presentation, throughout the rest of this paper we will use $\sigma$ for a list with a single state $\sigma$, and $\bar{\sigma}$ for a list with several states, which can also appear explicitly $\bar{\sigma} = [\sigma_1,\dots,\sigma_n]$. The semantics of a logic term, given an interpretation $I$ and a list of states is defined as follows:\\ \\
$|[ X|]_{I,\sigma} = \sigma(X)$\\
$|[ k|]_{I,\sigma} = I(k)$\\
$|[ f(t_1,\dots,t_n)|]_{I,\sigma}$ = {\bf if} $(domain(|[ t_1|]_{I,\sigma}),\dots,domain(|[ t_n|]_{I,\sigma})) \subseteq domain(I(f))$ \\
\hspace*{4cm} \textbf{then} $I(f)(|[ t_1|]_{I,\sigma},\dots,|[ t_n|]_{I,\sigma})$ \\
\hspace*{4cm} \textbf{else} $wrong$ \\

\noindent
The semantics for a predicate $p$ with arity $n$ corresponds to a function $f_p$, given by I, that given values from semantic domains $D_1 \times \dots \times D_n$ outputs values in $Bool$, or in case the values do not belong to the domains of the function, gives as error, $wrong$.\\

\noindent 
$|[ p|]_{I,\sigma} = I(p) = f_p$, where $f_p :: D_1 \times \dots \times D_n \rightarrow Bool$.
\\ \\
Given this, the semantics for programs is given as follows: \\

\noindent
$|[ t_1 = t_2|]_{I,\sigma}$~=~{\bf if} ($|[ t_1|]_{I,\sigma} = |[ t_2|]_{I,\sigma} \wedge |[t_1|]_{I,\sigma} \neq wrong$)\\
\hspace*{3cm} \textbf{then} $true$\\
\hspace*{3cm} \textbf{else if} $(domain(|[ t_1|]_{I,\sigma}) = domain(|[ t_2|]_{I,\sigma}) \wedge |[t_1|]_{I,\sigma}\neq wrong$ ) \\
\hspace*{5cm} \textbf{then} $false$\\
\hspace*{5cm} \textbf{else} $wrong$

\noindent
$|[ p(t_1,\dots,t_n)|]_{I,\sigma}$~=~{\bf if} $(domain(|[ t_1|]_{I,\sigma}),\dots,domain(|[ t_n|]_{I,\sigma})) \subseteq domain(|[p|]_{I,\sigma})$ \\
\hspace*{4cm} \textbf{then} $|[ p|]_{I,\sigma}(|[ t_1|]_{I,\sigma},\dots,|[ t_n|]_{I,\sigma})$ \\
\hspace*{4cm} \textbf{else} $wrong$ 

\noindent
$|[ g_1 , \dots ,g_n|]_{I,\sigma}~=~|[ g_1|]_{I,\sigma} \wedge \dots \wedge |[ g_n|]_{I,\sigma}$

\noindent
$|[ sg_1 ; \dots ; sg_m|]_{I,[\sigma_1,\dots,\sigma_m]}~=~|[ sg_1|]_{I,\sigma_1} \vee \dots \vee |[ sg_m|]_{I,\sigma_m}$

\noindent
$|[ q(X_1,\dots,X_n) :- sg_1;\dots;sg_m.|]_{I,[\sigma_1,\dots,\sigma_m]}~=~(|[ sg_1;\dots;sg_m|]_{I,[\sigma_1,\dots,\sigma_m]} \implies$ \\ \hspace*{5cm}$(|[ q(X_1,\dots,X_n)|]_{I,[\sigma_1]} \wedge \dots \wedge |[q(X_1,\dots,X_n)|]_{I,[\sigma_m]}))$
\newline

Note that conjunction, disjunction and implication in the previous definitions are interpreted in the three-valued logic defined by the truth tables previously presented in this section. Also note that, as different states are only needed for disjunctions, in the previous rules, the number of states in the list of states is one, except for the last two cases. 

Next function, called $or\_degree$, gives the number of states needed for the semantics of disjunctions.

\begin{defn}[($or\_degree$)]
Let $M$ be a term, a goal or a clause. Its $or\_degree$ is defined as follows:
\begin{itemize}
\item $or\_degree(M)$ = k, if $M = sg_1;\dots;sg_k$ or $M = p(X_1,\dots,X_n) :- sg_1;\dots;sg_k.$
\item $or\_degree(M)$ = 1, otherwise.
\end{itemize}

\end{defn}

\section{Types}
Here we define a new class of expressions, which we shall call {\em types}, build from an infinite enumerable set of type variables, a finite set of base types, an infinite and enumerable set of type constants and an infinite and enumerable set of type function symbols. 
Simple types can be:
\begin{itemize}
\item a type variable ($\alpha, \beta, \gamma, \dots$)
\item a type constant ($1, [~], `c\text{'},\dots \in TCons$)
\item a base type ($int, float, \dots \in TBase$)
\item a sum of simple types ($\tau_1 + \dots + \tau_n$)
\item a recursive type definition, where $\tau$ is a simple type and $\alpha$ occurs in $\tau$ ($\mu\alpha .\tau$)
\item a type function symbol $f \in TFunc$ associated with an arity n applied to an n-tuple of simple types ($f(1,[~], g(\alpha))$).

\end{itemize}
A {\em predicate type} is a functional type from a tuple of simple types to the type $bool$ ($\tau_1 \times \dots \times \tau_n \rightarrow bool$).
For recursive types we use a recursive operator $\mu$.
 For example, the traditional type for lists of integers can be written as: $\mu \alpha. ([~] + [int | \alpha])$.

Our type language enables parametric polymorphism through the use of type schemes. A {\em type scheme} is defined as $\forall \alpha. \tau$, where $\tau$ is a predicate type and $\alpha$ is a type variable.
In logic programming, there have been several authors that have dealt with polymorphism with type schemes or in a similar way \cite{DBLP:conf/slp/PyoR89,DBLP:journals/scp/BarbutiG92,Henglein:1993:TIP:169701.169692,DBLP:conf/iclp/Zobel87,DBLP:conf/lics/FruhwirthSVY91,DBLP:conf/iclp/GallagherW94,DBLP:books/mit/pfenning92/YardeniFS92,TAT}. Type schemes have type variables as generic place-holders for ground types. Parametric polymorphism comes form the fact these type variables (the type parameters) can be instantiated with any type. In the rest of the paper, for the sake of readability, we will not write the universal quantifiers on type schemes and we will assume that all free type variables on predicate types are universally quantified.

\subsection{Semantic Typing}

Here we need to define what is meant by a value $v$ semantically having a type $\tau$. We begin with a formal definition of a semantics for types. The notation we will use for a type constant is $\kappa$ and for a base type is $bs$. A simple type is {\em ground} if it contains no type variables and it is {\em complex} if it starts with a type function symbol $f$.

We assume that each base type $bs$ is associated with a basic domain. Therefore, not only there are exactly as many base types as there are basic domains, meaning that there is a finite number of base types, but we also know the association between them. Let $\sim$ denote the association between types and domains. If a base type $bs$ is associated with a basic domain $B_i$, we will denote this by $bs \sim B_i$. The association between base types and basic domains is considered to be predefined.

We will write $fix(F)$ meaning the least fixed point of function $F$, defined as $fix(F) = \bigcup_n F^n(\emptyset)$, where $F^i$ is the $i$-fold composition of $F$ with itself. Note that $fix(F)$ is well-defined, because the set of sets of values ordered by set inclusion is a complete partial order (CPO), and then, by Kleene fixed-point theorem, every function on this set has a least fixed point, which is the supremum of the ascending Kleene chain of the function. 

The relation $\sim$ can now be extended to relate types with domains in the following way:

\begin{itemize}
\item $bs \sim D$ is predefined.
\item $\tau_1 + \tau_2 \sim D_1 \cup D_2 \iff \tau_1 \sim D_1 \wedge \tau_2 \sim D_2$.
\item $\tau_1 \times \dots \times \tau_n \sim D_1 \times \dots \times D_n \iff \tau_1 \sim D_1 \wedge \dots \wedge \tau_n \sim D_n$.
\item $\alpha \sim D$, for any \emph{basic} or \emph{tree domain} $D$.
\item $\mu \alpha . \tau \sim D \iff F(x) = \mathbf{T}|[\tau [\nicefrac{x}{\alpha}] |]_I \wedge D = fix(F)$.
\item $\tau_1 \times \dots \times \tau_n \rightarrow \tau_{n+1} \sim D_1 \times \dots \times D_n \rightarrow D_{n+1} \iff \forall 1 \leq i \leq n+1.~\tau_i \sim D_i$
\end{itemize}

Note that the $\sim$ relation associates base types with basic domains and it is then lifted from this association of base types with basic domains. Thus, there is no way a type may be related by $\sim$ to the domain ${\bf W}$, containing the value {\em wrong}. 

Given the formerly described relation $\sim$, the semantics for types is given by the following rules. $\mathbf{T}|[~|]$ defines the semantics of types of terms, with the exception of $bool$, which is the type of the output of a predicate semantics. $\mathbf{P}|[~|]$ defines the semantics of types of predicates.\\ 

\noindent
$\mathbf{T|[ }\kappa \mathbf{|]}_I = \{ \kappa \}$\\
$\mathbf{T|[ }\alpha \mathbf{|]}_I = \{ v~|~v \in D \wedge \alpha \sim D \}$\\
$\mathbf{T|[ }bs \mathbf{|]}_I = \{v~|~v \in B_i \wedge bs \sim B_i \}$\\
$\mathbf{T|[ }bool \mathbf{|]}_I = \{ true ,~false \}$\\
$\mathbf{T|[ }f(\tau_1,\dots,\tau_n)\mathbf{|]}_I = \{f(v_1,\dots,v_n)~|~v_1 \in \mathbf{T|[ }\tau_1\mathbf{|]}_I \wedge \dots \wedge v_n \in \mathbf{T|[ }\tau_n\mathbf{|]}_I\}$\\
$\mathbf{T|[ } \tau_1 + \dots + \tau_n \mathbf{|]}_I = \mathbf{T|[ } \tau_1 \mathbf{|]}_I \cup \dots \cup \mathbf{T|[ } \tau_n \mathbf{|]}_I$.\\
$\mathbf{T|[ }\tau_1 \times \dots \times \tau_n \mathbf{|]}_I = \{ (v_1,\dots,v_n)~|~v_1 \in \mathbf{T}|[ \tau_1|]_I \wedge \dots \wedge v_n \in \mathbf{T}|[ \tau_n|]_I \}$.\\
$\mathbf{T|[ } \mu \alpha.\tau \mathbf{|]}_I = \{v ~|~v \in D \wedge \mu \alpha. \tau \sim D\}$

\noindent
$\mathbf{P|[ }\tau_1 \times \dots \times \tau_n \rightarrow bool \mathbf{|]}_I =$ \\ \hspace*{1cm} $\{p~|~\forall (t_1,\dots,t_n). |[ (t_1, \dots, t_n)|]_{I,\sigma} \in T|[ \tau_1 \times \dots \times \tau_n|]_I \implies |[ p|]_{I,A}(|[ t_1|]_{I,\sigma},\dots,|[ t_n|]_{I,\sigma}) \in \mathbf{T}|[bool|]_I \}$ \\
$\mathbf{P|[ }\forall \alpha. \tau \mathbf{|]}_I = \bigcap_{\tau\prime}. \mathbf{P}|[ \tau [\nicefrac{\tau\prime}{\alpha}|]_I$, where $\tau\prime$ is any ground simple type.

\begin{example}
Let us represent the type list of integers as $\mu \alpha. ([~] + [int~|~\alpha])$. Its semantics is calculated in the following way:\\ \\
$\mathbf{T}|[\mu \alpha. ([~] + [int~|~\alpha])|]_I ~=~ \{v~|~v \in D \wedge \mu \alpha. ([~] + [int~|~\alpha]) \sim D\}$\\ \\
Let $F(X)~=~\mathbf{T}|[[~] + [int~|~X])|]_I~=~\{[~] \} \cup \{ [n~|~v]~|~n \in D_1 \wedge v \in X \}$, where $int \sim D_1$. \\ \\
$F(\emptyset)~=~\{ [~] \} \cup \{[n~|~v]~|~n \in D_1 \wedge v \in \emptyset \}~=~\{ [~] \}$. \\
$F^2(\emptyset)~=~F(F(\emptyset))~=~F(\{[~]\})~=~\{[~]\} \cup \{[n]~|~n \in D_1 \}$\\
$F^3(\emptyset)~=~F(\{[~]\} \cup \{[n]~|~n \in D_1 \})~=~\{[~]\} \cup \{[n]~|~n \in D_1 \} \cup \{[n_1,n_2]~|~n_1 \in D_1 \wedge n_2 \in D_1\}$\\ \vdots \\
Now, note that $D~=~fix(F)$, where $fix(F)~=~\bigcup_n F^n(\emptyset)$.\\
Therefore, $D~=~F(\emptyset) \cup F^2(\emptyset) \cup F^3(\emptyset) \cup \dots~=~\{[~]\} \cup \{[n]~|~n \in D_1 \} \cup \{[n_1,n_2]~|~n_1 \in D_1 \wedge n_2 \in D_1\} \cup \dots $, where $D_1 \sim int$, which is the set of all finite lists of integers.

\end{example}

\begin{example}
Now consider the predicate type for a predicate defining polymorphic lists: $\forall \beta. \Big[ \mu \alpha.([~] + [\beta~|~\alpha]) \rightarrow bool \Big]$. The semantics of this type is the set of predicates which define lists of elements of type $\tau$, for every ground instance $\tau$ of $\beta$. Therefore its semantics contains all predicates whose semantic functions have polymorphic lists as the argument.
\end{example}

We shall now define what is meant by a value $v$ semantically having a type $\tau$. Note that values may have many types, or have no type at all. For example, the value {\em wrong} has no type.

Using {\bf T} we can define, for each state $\sigma$, a relation between the value associated with a variable in $\sigma$ and a type $\tau$, by:

\begin{defn}\label{JFS}
Let $X$ be a variable and $\tau$ a type:\\
$X:_{I,\sigma}\tau \iff |[ X|]_{I,\sigma} \in \mathbf{T|[ }\tau\mathbf{|]}_{I}$
\end{defn}\\

An {\em assumption} is a type declaration for a variable, written $X:\tau$, where $X$ is a variable and $\tau$ a type. Variable $X$ is called the {\em subject} of the assumption.
We define a {\em context} $\Gamma$ as a set of assumptions with distinct variables as subjects (alternatively {\em contexts} can be defined as functions from variables to types). We can extend the above relation to contexts.

\begin{defn}\label{ContextSemantics}
Given a context $\Gamma$,
$|[ \Gamma|]_{I,\sigma} \iff \forall (X:\tau)\in \Gamma.~ X:_{I,\sigma}\tau$
\end{defn}

\begin{defn}\label{OPLUS}
 The {\em sum of contexts}, written $\Gamma_1 \oplus\Gamma_2$, is defined by: \\ 
$(\Gamma_1 \oplus \Gamma_2)(X)=
\begin{cases}
\Gamma_1(X), &   X \text{ does not occur as subject in } \Gamma_2\\
\Gamma_2(X), &   X \text{ does not occur as subject in } \Gamma_1\\
\Gamma_1(X) + \Gamma_2(X), & \text{otherwise.}
\end{cases}$
\end{defn}\\

Finally we give a semantic meaning to assertions of the
form $\Gamma \models_{P,I} M:\tau$ stating that if the assumptions in $\Gamma$ hold then $M$ yields a value of type $\tau$.

\begin{defn}[(Semantic Typing)]\label{STR}
Let $P$ be a program, i.e. a set of clauses, $I$ an interpretation, and $M$ either a term, a goal or a clause.\\
$\Gamma \models_{P,I} M:\tau \iff \exists [\Gamma_1,\dots,\Gamma_n]. \forall \bar{\sigma} = [\sigma_1,\dots,\sigma_n]. \Big[ |[ \Gamma_1|]_{I,[\sigma_1]} \wedge \dots \wedge |[ \Gamma_n|]_{I,[\sigma_n]} \implies \mathbf{|[ }M|]_{I,\bar{\sigma}} \in \mathbf{T}|[\tau|]_I \Big]$,\\
where $\Gamma_1 \oplus \dots \oplus \Gamma_n = \Gamma$ and $n = or\_degree(M)$. If $n = 1$, then the only sum possible is $\Gamma$ itself.
\end{defn}

\begin{example}
Let $p$ be a predicate with the following predicate definition:
\begin{verbatim}
p(X) :- X = 1 ; X = a.
\end{verbatim}
Let interpretation $I$ be such that $I(\texttt{1}) = 1$ and $I(\texttt{a}) = a$ and $B_1$ and $B_2$ be two semantic domains such that $1 \in B_1$ and $a \in B_2$. Let $I(\texttt{p}) = f_p$, such that $f_p :: B_1 \cup B_2 \to Bool$.
Lets assume we have $\Gamma = \{ X : int + atom \}$, $\Gamma_1 = \{ X : int\}$ and $\Gamma_2 = \{ X : atom \}$, where $int \sim B_1$ and $atom \sim B_2$. We will show that $\Gamma \models_{P,I} (p(X) :- X = 1 ; X = a.) : bool$. This corresponds to showing that $$\exists [\Gamma_1,\Gamma_2] .\forall [\sigma_1,\sigma_2]. \Big[ |[\Gamma_1|]_{I,[\sigma_1]} \wedge |[\Gamma_2|]_{I,[\sigma_2]}  \implies \mathbf{|[ }p(X) :- X = 1 ; X = a.|]_{I,\bar{\sigma}} \in \mathbf{T}|[bool|]_I \Big]$$
Suppose $\Gamma_1 = \{X : int\}$ and $\Gamma_2 = \{X:atom\}$, then $\Gamma_1 \oplus \Gamma_2 = \Gamma$. If $\sigma_1(X) \in B_1$ and $\sigma_2(X) \in B_2$, the left hand side of the implication is true. The right hand side is also true, since applying $f_p$ to any of the $\sigma_i(X)$ does not return \emph{wrong} and neither does any of the unifications on the bodies of the clause. Therefore the semantic value of the clause is not \emph{wrong}. 
If one of the $\sigma_i(X)$ does not yield a value in the previous domains, the right hand side of the implication is false, since one of the unifications yields \emph{wrong}. But the left hand side is also false, since $|[X|]_{I,\sigma_1} \notin \mathbf{T}|[int|]_I$ and $|[X|]_{I,\sigma_2} \notin \mathbf{T}|[atom|]_I$, thus the initial statement is still true.
\end{example}

\section{Type System}

In this section we define a type system which, statically, relates logic programs with types. 
The type system defines a relation $\Gamma \vdash_{P} p:\tau$, where $\Gamma$ is a context as defined in the previous section, $p$ is a term, a goal, or a clause, and $\tau$ is a type. This relation should be read as expression $p$ has type $\tau$ given the context $\Gamma$, in program $P$. We will write $\Gamma \cup \{X:\tau\}$ to represent the context that contains all assumptions in $\Gamma$ and the additional assumption $X:\tau$ (note that because each variable is unique as a subject of an assumption in a context, in $\Gamma \cup \{X:\tau\}$, $\Gamma$ does not contain assumptions with $X$ as subject). We will write a sequence of variables $X_1,\dots, X_n$ as $\bar{X}$, and a sequence of types as $\bar{\tau}$.
We assume that clauses are normalized and, therefore, every call to a predicate in the body of a clause contains only variables.

The type system also uses a function $type$ which gives the type of constants and function symbols. We assume that $type$ is defined for all constants and function symbols that occur in $P$ and types given by $type$ never include $bool$ and always have the right arity, i.e. for a function symbol of arity $n$ it will be of the form $\tau_1 \times \dots \times \tau_n \rightarrow \tau\prime$. If $n=0$, it is just $\tau$.
This function corresponds to a type declaration. For instance, the list type $list(\alpha) = [~] + [\alpha~|~list(\alpha)]$ is denoted by $\mu \beta. [~] + [\alpha~|~\beta]$ in the type system. This is defined in the type system by a function $type$ that assigns the type $\mu \beta. [~] + [\alpha~|~\beta]$ to the empty type $[~]$ and the type $\alpha \times (\mu \beta. [~] + [\alpha~|~\beta]) \to (\mu \beta. [~] + [\alpha~|~\beta])$ to the type constructor symbol $[~|~]$.
Note that these declarations for function symbols and constants may be used, however, they are optional and, moreover, there are no type declarations for predicates.

We first define the following subtyping relation, which will be used in the type system.

\begin{defn}[(Subtyping)]\label{Subtyping}
Let $\phi$ be a substitution of types for type variables. Let $\leq$ denote the subtyping relation defined as follows:
\begin{itemize}
\item $\tau \leq \tau$ (Reflexivity)

\item $\tau \leq \tau\prime$ if $\exists \phi. \phi(\tau\prime) = \tau$ (Instance)

\item if $\tau \leq \tau\prime$, then $\tau \leq \tau\prime + \tau\prime\prime$ (Right Union)

\item if $\tau \leq \tau\prime$, then $\tau \leq \tau\prime\prime + \tau\prime$ (Left Union)

\item if $\tau\prime \leq \tau$, then $\tau \rightarrow bool \leq \tau\prime \rightarrow bool$ (Contravariance)
\end{itemize}

\end{defn}

\begin{lemma}[(Soundness of the subtyping relation)]
Given a predicate $p$ and a type $\tau$, if $p \in \mathbf{P}|[ \tau|]_{I}$, then $\forall \tau\prime. \tau \leq \tau\prime \implies p \in \mathbf{P}|[ \tau\prime|]_{I}$.
\end{lemma}

\begin{figure}[!ht]
\[
\begin{array}{ll}
VAR & \Gamma \cup \{X:\tau\} \vdash_{P} X:\tau \\ \\ \\
CST & \Gamma \vdash_{P} c:\tau, \text{ where } type(c) = \tau\\ \\ \\
CPL & \frac{{ \textstyle \Gamma \vdash_{P} t_1: \tau_1 ~~\dots ~~ \Gamma \vdash_{P} t_n:\tau_n}}{{ \textstyle \Gamma \vdash_{P} f(t_1,\dots,t_n) : \tau }}, \text{ where } type(f) = \tau_1 \times \dots \times \tau_n \rightarrow \tau \\ \\ \\
UNF & \frac{{\textstyle \Gamma \vdash_{P} t_1:\tau ~~~  	\Gamma \vdash_{P} t_2:\tau}}{{\textstyle \Gamma \vdash_{P} t_1=t_2 : bool}}   \\ \\ \\
CLL & \frac{{\textstyle \Gamma \cup \{Y_1 :\tau_1\prime, \dots , Y_n : \tau_n\prime\} \vdash_{P} (p(Y_1, \dots ,Y_n):-sg.) : bool ~~~~ \forall i. \tau_i \leq \tau_i\prime}}{{\textstyle  \Gamma \cup \{X_1:\tau_1, \dots, X_n:\tau_n\} \vdash_{P} p(X_1,\dots ,X_n):bool}} \\ \\ \\
CON & \frac{{\textstyle \Gamma \vdash_{P} g_1:bool ~~  \Gamma \vdash_{P} g_n:bool}}{{\textstyle \Gamma \vdash_{P} g_1,\dots,g_n :bool}}   \\ \\ \\
CLS^{(a)} & \frac{{\textstyle \Gamma \cup \{\bar{X}:\bar{\tau_1}\} \vdash_{P} b_1:bool ~~ \dots ~~ \Gamma \cup \{\bar{X}:\bar{\tau_m}\} \vdash_{P} b_m:bool }}{{\textstyle \Gamma \cup \{\bar{X}:\bar{\tau_1} + \dots + \bar{\tau_m} \} \vdash_{P} (p(\bar{X}):- b_1;\dots;b_m.) : bool }} \\ \\ \\
RCLS^{(b)} & \frac{\textstyle \Gamma \cup \{\bar{X}:\bar{\tau_i}\} \vdash_{P} b_i:bool ~~ \dots ~~  \Gamma \cup \{\bar{X}:\bar{\tau_j}, \bar{Y}_{j1}:\bar{\tau_j}, \dots, \bar{Y}_{jk_j}:\bar{\tau_j}\} \vdash_{P} b_{m+j}:bool}{\textstyle \Gamma \cup \{ \bar{X}: \bar{\tau_i} + \bar{\tau_j}\} \vdash_{P} (p(\bar{X}):- b_1;\dots;b_m; \hspace*{3cm}} \\ & \hspace*{5.1cm} { \small b_{m+1},p(\bar{Y}_{11}), \dots, p(\bar{Y}_{1k_1});} \\ & \hspace*{7cm} {\small \vdots} \\ & \hspace*{5.1cm} {b_{m+n},p(\bar{Y}_{n1}), \dots, p(\bar{Y}_{nk_n}).) : bool} \\ \\ 
\end{array}
\]
\scriptsize{(a) This rule is for non-recursive predicates only. The sum on the consequence of the rule is argument-wise}
\scriptsize{(b) This rule is for recursive predicates. Note that all variables in recursive calls in a certain sequence of goals have the same type as the variables in the head in that sequence of goals. Also $i = 1,\dots,m$ and $j=1,\dots,n$}
\caption{Type System}
\label{TSystem}
\end{figure}

Then we present a type system (see Figure \ref{TSystem}) defining the typing relation, which relates terms, calls and predicate definitions with types. If there is a context $\Gamma$ and a type $\tau$ such that $\Gamma \vdash_{P} p:\tau$ we say that $p$ is {\em statically well-typed}. This type system can be easily implemented to type check programs, but not to infer types. Nevertheless, we want to keep open the possibility of defining a type inference algorithm.

\begin{defn}[(Monomorphism Restriction)] \label{MR}
Let $p$ be a recursive predicate of arity $n$, typed with type $\tau$ using an context $\Gamma$. For all $1 \leq i \leq n$, the types of the variables in the $i$-th argument of $p$ in the head of the clause defining it, and in all of its recursive calls are the same in $\Gamma$.
\end{defn}\\

It is well-known that type inference in the presence of polymorphic recursion is not decidable \cite{Henglein:1993:TIP:169701.169692,Kfoury:1993:TRP:169701.169687}, thus we do not allow polymorphic recursion in the system. This is achieved by the previous restriction on recursive predicates.
We choose to define this restriction locally in each predicate definition for the sake of simplicity of presentation. The alternative would be to define a new syntax for logic programming to group together mutually recursive predicates as a single syntactic entity (in functional programming this would correspond to nested {\em letrec} expressions). The {\em monomorphism restriction} (Definition \ref{MR}) holds in our type system by rule RCLS for typing recursive predicates. In this rule we use the same type for the variables in the head of the clause in its recursive calls in the body.

Let us briefly describe the other type rules. Rule Var types a variable with the type is has in the context.

Rule CST and CPL types constants and complex terms using the $type$ function either directly in the case of constants, or checking the types for the arguments in the case of a complex term, which have to be the same ones as the input of $type(f)$, for a complex term $f(t_1,\dots,t_n)$.

Rule UNF types an equality as $bool$ if the types for both sides of the equality are the same.
Rule CLL types predicate calls: for a call to a predicate $p$ to be well typed, the type for each variable in the call needs to be a subtype of the type of the variables in the definition of $p$ in program $P$.

Rule CON just check that every goal is $bool$.
Rule CLS types non-recursive clauses: if we type each body of a clause using types $\bar{\tau}$ for variables $\bar{X}$, then we can type the entire clause with the sum of all those types argument-wise.
Note that, from rules CLS and RCLS, the type of a clause is $bool$. However, the interesting type information is the type for a predicate, determined by the types of its arguments which, in the end of the type derivation, are in context $\Gamma$.

\begin{example}
Suppose that the programmer wants to use a list data structure. Then it can declare it in the program as such:\\ \\
$:- type$ $list(\alpha)=[~] + [\alpha~|~list(\alpha)].$\\ \\
This declaration is translated to the following in the type system: $\mu \beta. ([~] + [\alpha~|~\beta])$. Then, every constant and constructor that are members of the list of possible cases for the type and are assigned their corresponding types. The function $type$ in the type system will be:\\ \\
$type([~]) = \mu \beta.( [~] + [\alpha~|~\beta])$,\\
$type([~|~]) = \alpha \times \mu \beta. ([~] + [\alpha~|~\beta]) \to \mu (\beta. [~] + [\alpha~|~\beta])$.\\
\\
With this $type$ function, we can produce a derivation that assigns the type $list(\alpha) \times list(\alpha) \times list(\alpha) \to bool$ to the predicate $append$ defined as usual.\\
In this case, the call $append([~],1,1)$ would be ill-typed. 

Note that with another function $type$, originated from another declaration from the programmer, the same call could be well-typed. One example would be the following declaration:\\ \\
$:- type$ $dummy(\alpha)= 1 + [~] + [\alpha~|~dummy(\alpha)]$.
\end{example}

\begin{example}
Here we give a simple example of a type derivation.
Let $p$ be a predicate defined by {\em p(X) :- X = 1; X = a.}. Let $type$ be such that $type(1) = int$ and $type(a) = atom$.
By two applications of rule UNF followed by an application of rule CLS we have: \\ \\
\noindent
\begin{math}
\begin{array}{c}
\frac{{\textstyle \{X: int\} \vdash_{P} X: int ~~~~ \{X: int\} \vdash_{P} 1: int }}{{\textstyle \{X: int\} \vdash_{P} X = 1: bool }}
\end{array}
\end{math}
\\ \\

\noindent
\begin{math}
\begin{array}{c}
\frac{{\textstyle \{X: atom\} \vdash_{P} X : atom ~~~~ \{X: atom\} \vdash_{P} a: atom}}{{\textstyle \{X: atom\} \vdash_{P} X = a: bool }}
\end{array}
\end{math}
\\ \\

\noindent
\begin{math}
\begin{array}{c}
\frac{{\textstyle \{X: int\} \vdash_{P} X=1: bool ~~~~ \{X: atom\} \vdash_{P} X = a: bool }}{{\textstyle \{X: int + atom\} \vdash_{P} p(X) :- X = 1; X = a. : bool }}
\end{array}
\end{math}
\\ \\
From the type of $X$ in the final context, the type of $p$ is $int + atom \to bool$.
\end{example}

We will also give an example for a predicate with a type declaration.

\begin{lemma}[(Interpretation Existence)]\label{IE}
Given a function \emph{type}, there is always an interpretation $I$, such that for any constant or function symbol, here denoted by symbol $t$, and any state $\sigma$, $|[t|]_{I,\sigma} \in \mathbf{T}|[type(t)|]_I$.
\end{lemma}

Finally, our main result shows that the type system is semantically sound, meaning that if a program has a type in our type system, then the program and its type are related by the semantic typing relation defined in Definition \ref{STR}.

\begin{theorem}[(Semantic Soundness)]\label{SS}
Let $P$ be a program, then $\Gamma \vdash_{P} M:\tau \implies \exists I. (\Gamma \models_{P,I} M:\tau)$
\end{theorem}

Note that the value \emph{wrong} has no type, thus, as a corollary of the soundness theorem, we have that if a predicate is statically well-typed, then there is an interpretation $I$, for which the predicate semantics is not \emph{wrong}.

\section{Conclusion}

Here we present a new semantics for logic programming which captures the notion of type error with another value in the logic itself. For a restricted language we have shown that this semantics can be used to prove that a type system for logic programming is correct based on a notion that connects the semantics of programs and the semantics of types.

The next step, left for future work, will be the definition of a type inference algorithm, sound with respect to the type system presented in this paper, that automatically infers types for programs and may use type declaration in the form of data structures, although that is optional. 
\\

{\small {\bf Acknowledgments} This work is partially funded by FCT within project Elven
POCI-01-0145-FEDER-016844, Project 9471 - Reforcar a
Investigacao, o Desenvolvimento Tecnologico e a Inovacao
(Project 9471-RIDTI), by project PTDC/EEI-CTP/3506/2014, and by Fundo Comunitario Europeu
FEDER.}

\bibliographystyle{eptcs}

\newpage
\appendix
\setcounter{thm}{0}

\begin{lemma}[(Soundness of the subtyping relation)]\label{TSS}
Given a predicate $p$ and a type $\tau$, if $p \in \mathbf{P}|[ \tau|]_{I}$, then $\forall \tau\prime. \tau \leq \tau\prime \implies p \in \mathbf{P}|[ \tau\prime|]_{I}$.
\end{lemma}

\begin{proof} The proof follows directly from the definition of the $\leq$ relation and the definition of $\mathbf{P}|[~|]$.
\end{proof}

\begin{lemma}[(Interpretation Existence)]
Given a function \emph{type}, there is always an interpretation $I$, such that for any constant or function symbol, here denoted by symbol $t$, and any state $\sigma$, $|[t|]_{I,\sigma} \in \mathbf{T}|[type(t)|]_I$.
\end{lemma}

\begin{proof}
The proof follows from the definition of function $\mathbf{T|[~|]}$, noting that, as there is an infinite number of possible interpretation functions $I$, it is always possible to choose an appropriate interpretation.
\end{proof}

\begin{theorem}[(Semantic Soundness)]
Let $P$ be a program, then $\Gamma \vdash_{P} M:\tau \implies \exists I. (\Gamma \models_{P,I} M:\tau)$
\end{theorem}

\begin{proof}[Proof of Theorem \ref{SS} (Semantic Soundness):]
The proof of this theorem follows by structural induction on $M$.\\

\noindent Base cases:

\begin{itemize}
\item VAR: We know that $X:\tau \in \Gamma$ and we want to prove that there is an interpretation $I$ such that $\Gamma \models_{P,I} X:\tau$, which corresponds to proving $\exists I. \forall \sigma. \Big[ |[ \Gamma|]_{I,\sigma} \implies \mathbf{|[ }X|]_{I,\sigma} \in \mathbf{T}|[ \tau|]_I \Big]$. For any $I$, suppose for some $\sigma$, $|[ \Gamma|]_{I,\sigma}$ is false, then the implication is true and we get the result we wanted. Suppose for some other $\sigma$, $|[ \Gamma|]_{I,\sigma}$ is true. Then it follows by Definition \ref{ContextSemantics} that the right side of the implication is true, which means the whole implication is also true.

\item CST: We want to prove that $\exists I. \Gamma \models_{P} c:\tau$, which corresponds to proving $\exists I. \forall \sigma. \Big[ |[ \Gamma|]_{I,\sigma} \implies \mathbf{|[}c|]_{I,\sigma} \in \mathbf{T}|[ \tau|]_I \Big]$. From lemma \ref{IE}, there is always an interpretation $I$ such that the right side of the implication is true. Therefore on that interpretation $I$, the implication is true for any $\sigma$.
\end{itemize}

\noindent Inductive Step:

\begin{itemize}
\item CPL: We want to prove that $\exists I. (\Gamma \models_{P,I} f(t_1,\dots , t_n) : \tau)$, which corresponds to proving the following: $\exists I. \forall \sigma. \Big[ |[ \Gamma|]_{I,\sigma} \implies \mathbf{|[ }f(t_1, \dots , t_n)|]_{I,\sigma} \in \mathbf{T}|[ \tau|]_I \Big]$. By the induction hypothesis, we know that $\exists I_i. \Gamma \models_{P,I_i} t_i:\tau_i, \forall 1 \leq i \leq n$. Let $I$ be an interpretation such that $|[t_i|]_{I_i,\sigma} = |[t_i|]_{I,\sigma}$. For that $I$, for any state $\sigma$, if $domain(I(f)) ~ \tau_1 \times \dots \times \tau_n \rightarrow \tau$ then the implication is true. If $domain(I(f))$ is not associated with the correct type, then we can always create an interpretation $I\prime$ that is the same as $I$ for all constants, function symbols and predicate symbols, but differ on $f$, such that the relation $domain(I(f)) ~ \tau_1 \times \dots \times \tau_n \rightarrow \tau$ holds. For all states $\sigma\prime$ that make the left side of the implication false the implication is trivially true. Therefore, for that $I$ or $I\prime$, for any $\sigma$, $|[ f(t_1, \dots , t_n)|]_{I,\sigma} \in \mathbf{T}|[ \tau|]_I$.

\item UNF: We want to prove $\exists I. \Gamma \models_{P,I} t_1 = t_2 : bool$, which corresponds to proving $\exists I.\forall \sigma. \Big[ |[ \Gamma|]_{I,\sigma} \implies \mathbf{|[ }t_1 = t_2|]_{I,\sigma} \in \mathbf{T}|[ bool|]_I \Big]$. By the induction hypothesis, we know that $\exists I_1.\Gamma \models_{P,I_1} t_1:\tau$ and $\exists I_2. \Gamma \models_{P,I_2} t_2:\tau$. Let $I$ be an interpretation constructed from $I_1$ and $I_2$ in such a way that, for any $\sigma$, $|[t_1|]_{I_1,\sigma} = |[t_1|]_{I,\sigma}$ and $|[t_1|]_{I_2,\sigma} = |[t_1|]_{I,\sigma}$. In this $I$, we know that $|[t_1|]_{I,\sigma}$ and $|[t_2|]_{I,\sigma}$ belong to the same semantic domain since there is only one domain associated with $\tau$, which means the semantics of $t_1 = t_2$ is never $wrong$, i.e. is in $\mathbf{T}|[bool|]_I$. Therefore, for any $\sigma$, $\Gamma \models_{P,I} t_1 = t_2 : bool$.

\item CLL: We want to prove $\exists I. \Gamma \cup \{X_1:\tau_1, \dots, X_n:\tau_n\} \models_{P,I} p(X_1,\dots, X_n) : bool$, which corresponds to proving $\exists I. \forall \sigma. \Big[ |[ \Gamma|]_{I,\sigma} \implies |[p(X_1,\dots, X_n)|]_{I,\sigma} \in \mathbf{T}|[ bool|]_I \Big]$. By the induction hypothesis, we know that $\exists I. \Gamma \cup \{Y_1:\tau_1\prime, \dots, Y_n:\tau_n\prime\} \models_{P} p(Y_1,\dots,Y_n):- body. : bool$, which means $|[ p|]_{I,\sigma} \in \mathbf{P}|[ \tau_1\prime \times \dots \times \tau_n\prime \rightarrow bool|]_I$. Because $\forall i = 1,\dots,n .~\tau_i \leq \tau\prime_i$, then, by definition \ref{Subtyping}, we have that $\tau_1\prime \times \dots \times \tau_n\prime \to Bool \leq \tau_1 \times \dots \times \tau_n \to Bool$. Thus, from lemma \ref{TSS}, it follows that $|[ p|]_{I,\sigma} \in |[ \tau_1 \times \dots \times \tau_n \rightarrow bool|]_I$. Therefore, for I, for any $\sigma$, $\Gamma \models_{P,I} p(X_1,\dots,X_n) : bool$.

\item CON: We want to prove $\exists I. \Gamma \models_{P,I} g_1,\dots,g_n :bool$, which corresponds to proving $\exists I. \forall \sigma. \\ \Big[ |[ \Gamma|]_{I,\sigma} \implies |[g_1,\dots,g_n|]_{I,\sigma} \in \mathbf{T}|[ bool|]_I \Big]$. By the induction hypothesis, for $1 \leq i \leq n$, $\exists I_i. \Gamma \models_{P,I_i} g_i :bool$. Let $I$ be an interpretation such that $|[c_i|]_{I_i,\sigma} = |[c_i|]_{I,\sigma}$ for all $i$ and any $\sigma$. With this interpretation, $\forall \sigma. \Big[ |[ \Gamma|]_{I,\sigma} \implies \mathbf{|[ }g_1,\dots,g_n|]_{I,\sigma} \in \mathbf{T}|[ bool|]_I \Big]$.

\item CLS: We want to prove $\exists I. \Gamma \models_{P,I} p(X_1,\dots , X_n) :- b_1 ; \dots ; b_m. : bool$, meaning $\exists I.\exists [\Gamma_1,\dots,\Gamma_m]. \\ \forall \bar{\sigma} = [\sigma_1,\dots,\sigma_m] . \Big[ |[ \Gamma_1|]_{I,\sigma_1} \wedge \dots \wedge |[\Gamma_m|]_{I,\sigma_m} \implies \mathbf{|[ }p(X_1,\dots , X_n) :- b_1; \dots ; b_m.|]_{I,\bar{\sigma}} \in \mathbf{T}|[ bool|]_I \Big]$. By the induction hypothesis, we know that $\exists I_i. \Gamma_i \models_{P,I_i} b_i :bool$, for $1 \leq i \leq m$, where $\Gamma_i$ is the same for all variables except $\bar{X}$. The fact that $\Gamma_1 \oplus \dots \oplus \Gamma_m = \Gamma$ comes from the contexts being the same for all variables except $\bar{X}$ and from $\Gamma$ having the assumptions $\bar{X}:\bar{\tau_1}+\dots + \bar{\tau_m}$. We know that each $b_i$ is modelled by $I_i$ for all $\sigma$, then let $I\prime$ be an interpretation such that for all $\sigma$ and all $i$, $|[b_i|]_{I_i,\sigma} = |[b_i|]_{I\prime,\sigma}$. If $I\prime$ is such that $I\prime(p) = f :: D_1 \times \dots \times D_n \rightarrow D$ and $\tau_{1,i} + \dots + \tau_{m,i} \sim D_i$, for $1 \leq i \leq n$, then let $I = I\prime$, else we can create $I$ by changing only the semantic function of $p$ in $I\prime$, for an $f$ that is as described above. Then we have that $\Gamma \models_{P,I} p(X_1,\dots , X_n) :- b_1, \dots , b_m. : bool$.

\item RCLS: We want to prove $\exists I. \Gamma \models_{P,I} p(X_1,\dots , X_n) :- b_1, \dots , b_m. : bool$, meaning $\exists I.\exists [\Gamma_1,\dots \Gamma_k].\\  \forall \bar{\sigma} = [\sigma_1,\dots,\sigma_k] . \Big[ |[ (p(\bar{X}):- b_1;\dots;b_m; b_{m+1},p(\bar{Y}_{11}),\dots, p(\bar{Y}_{1k_1}); \dots; b_{m+n},p(\bar{Y}_{n1}), \dots, \\ p(\bar{Y}_{nk_n}). |]_{I,\bar{\sigma}} \in \mathbf{T}|[ bool|]_I \Big]$. By the induction hypothesis, we know that for $1 \leq i \leq m+n$, $\exists I_i. \Gamma_i \models_{P,I_i} b_i : bool$, where $\Gamma_i$ is of the form that occurs in the rule. We also know that $\Gamma_1 \oplus \dots \oplus \Gamma_{m+n} = \Gamma$, since the contexts are the same for all variables except $\bar{X}$ and the $\bar{Y_j}$ and we can add the assumptions for the $\bar{Y_j}$ whenever they do not exist already in a context with no change in the implication since those variables only occur in one sequence of goals in the body of the clause. For $\bar{X}$, we know that $\Gamma(\bar{X}) = \Gamma_1(\bar{X}) + \dots + \Gamma_{m+n}(\bar{X})$, so we prove the previous $\oplus$ statement. Let $I\prime$ be an interpretation such that $|[b_i|]_{I_i,\sigma} = |[b_i|]_{I\prime,\sigma}$ for all $i = 1, \dots, m+n$. If $I\prime$ is such that $I\prime(p) = f :: D_1 \times \dots \times D_n \rightarrow D$ and $\tau_{1,i} + \dots + \tau_{m,i} \sim D_i$, for $1 \leq i \leq n$, then let $I = I\prime$, else we can create $I$ by changing only the semantic function of $p$ in $I\prime$, for an $f$ that is as described above. Then we get that for $I$, $\Gamma \models_{P,I} p(X_1,\dots , X_n) :- b_1, \dots , b_m. : bool$.
\end{itemize}
 
\end{proof}

\end{document}